\documentclass[12pt]{amsart}
\usepackage{latexsym,amsmath,amsfonts,amscd,amssymb}
\usepackage{enumitem}
\usepackage{graphicx}
\newtheorem{theorem}{Theorem}[section]
\newtheorem{theorem*}{Theorem}
\newtheorem{lemma}[theorem]{Lemma}
\newtheorem{proposition}[theorem]{Proposition}
\newtheorem{proposition*}[theorem*]{Proposition}
\newtheorem{corollary}[theorem]{Corollary}
\newtheorem{corollary*}[theorem*]{Corollary}

\newtheorem{definition}[theorem]{Definition}

\theoremstyle{remark}

\newtheorem{remark*}[theorem*]{Remark}

\newtheorem{note*}[theorem*]{Note}

\newcommand{\EE}{{\mathbb E}}

\begin{document}

\title{On Profitability of Trailing Mining}


\author{Cyril Grunspan}
\address{L{\'e}onard de Vinci, P{\^o}le Universitaire, Research Center, Paris-La D{\'e}fense, Labex R{\'e}fi, France}
\email{cyril.grunspan@devinci.fr}

\author{Ricardo Perez-Marco}
\address{CNRS, IMJ-PRG, Labex R{\'e}fi, Paris, France}
\email{ricardo.perez.marco@gmail.com}
\address{\tiny Author's Bitcoin Beer Address (ABBA)\footnote{\tiny Send some anonymous and moderate satoshis to support our research at the pub.}:\newline{}\indent 1KrqVxqQFyUY9WuWcR5EHGVvhCS841LPLn} 

\address{\includegraphics[scale=0.5]{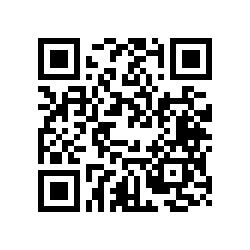}}


\begin{abstract} 
  We compute the revenue ratio of the Trail Stubborn mining strategy in the Bitcoin network and
  compare its profitability to other block-withholding strategies. We use for this martingale techniques 
  and a classical analysis of the hiker problem.
  In this strategy the attacker could find himself mining in a shorter fork, but 
  we prove that for some parameter values it is still profitable to not give up. 
  This confirms previous numerical studies.
\end{abstract}

{\maketitle}

\section{Introduction}

In our previous article \cite{GPM2018} we gave a rigorous foundation for the profitability  
of alternative mining strategies in the Bitcoin network \cite{N08}. As for games with repetition, 
it depends on the proper 
analysis of the revenue and the duration over attack cycles. 
More precisely, the expected revenue $\EE[R]$ and expected duration $\EE[\tau]$ 
over an attack cycle give the ``Revenue Ratio''
$$
\Gamma=\frac{\EE[R]}{\EE[\tau]}
$$
This is the correct benchmark for the profitability of the strategy.

\medskip

This analysis was also carried out previously by the authors in \cite{GPM2018} to the ``Selfish Mining'' 
(SM) strategy from \cite{ES14} and 
in \cite{GPM2018-2} to 
``Lead-Stubborn Mining'' (LSM) and ``Equal Fork 
Stubborn Mining'' (EFSM) strategies from \cite{NKMS2016}. In these articles we found 
for these strategies the exact mathematical formula for the 
revenue ratios and we compared  its profitability in parameter space.

The main technique for these derivations is the application of martingale techniques 
introduced in \cite{GPM2018} that yield, using Doob's Stopping Time Theorem, the expected 
duration of the attack cycles. We assume in the computation of the Revenue Ratio 
that there is no difficulty adjustment inside the attack cycles, or that 
$E[\tau]$ is much shorter than the period of difficulty adjustment so that its effect in the cycles 
can be neglected. But, as we proved in \cite{GPM2018}, the effect of these attacks is to slow down the 
network, and it is only after a 
difficulty adjustment that these ``block withholding strategies'' can become profitable. Then
the profitability can also be read on the apparent hashrate. So these rogue strategies are an exploit 
on the difficulty adjustment formula and we gave in \cite{GPM2018} an improvement proposal of the Bitcoin
protocol to fix the difficulty adjustment formula.

\medskip

In this article we apply again these new techniques 
to another block withholding strategy: The ``Trail-Stubborn Mining'' strategy from \cite{NKMS2016}.
In this strategy, the block withholder miner does not give up when the 
honest chain takes over some block advantage, but instead keeps mining on top of his secret 
chain and only gives up if the advantage of the honest chain reaches $A\geq 1$ blocks. 
We talk about ``$A$-Trail-Stubborn Mining'' strategy or ${\text{TSM}}_A$ in short.

We denote by $b>0$ be the block reward, and $\tau_0$ the average inter-block 
validation time for the total network (around $10$ minutes for the Bitcoin network). We denote 
by $q$ (resp. $p$) the relative hashing power of the attacker (resp. honest miners) and $\lambda=q/p<1$. Let 
$\gamma$ be the fraction of the honest network that the attacker attracts to mine on top of his fork. 
For a miner that after a difficulty adjustment has a Revenue Ratio $\tilde \Gamma$ we define 
his apparent hashrate $\tilde q$ by
$$
\tilde q= \frac{\tilde \Gamma \cdot \tau_0}{b} \ .
$$
The apparent hashrate of a miner can also be defined after a difficulty adjustment as the average
proportion of blocks mined by the miner in the official blockchain.
Our main Theorem is:

\begin{theorem*}[$A$-Trail-Stubborn mining] \label{thm_main}
  Let $A\geq 1$. The revenue ratio of the ``$A$-Trail-Stubborn mining'' strategy is 
 
  $$
  \Gamma  =  \frac{q + \frac{(1 - \gamma) pq (p - q)}{(p + pq - q^2) [A + 1]}  \left(
  \left( [A - 1] + \frac{1}{p}  \frac{P_A (\lambda)}{[A + 1]} \right) \lambda^2
  - \frac{2}{\sqrt{1 - 4 (1 - \gamma) pq} + p - q}  \right)}{1 + \frac{(1 -
  \gamma) pq}{p + pq - q^2}  (A + 1)  \left( \frac{[2]}{[A + 1]} - \frac{2}{A +
  1} \right)}\, \frac{b}{\tau_0}
  $$
  
  where $[n]=\frac{1-\lambda^n}{1-\lambda}$ for $n \in \mathbb{N}$, 
  and $P_A (\lambda) = \frac{1 - A \lambda^{A - 1} + A \lambda^{A + 1} - \lambda^{2 A}}{(1 - \lambda)^3}$.

After a difficulty adjustment, the apparent hashrate  of the stubborn miner is
  $$
  \tilde{q}  =  \frac{q + \frac{(1 - \gamma) pq (p - q)}{(p + pq - q^2) [A + 1]}  \left(
  \left( [A - 1] + \frac{1}{p}  \frac{P_A (\lambda)}{[A + 1]} \right) \lambda^2
  - \frac{2}{\sqrt{1 - 4 (1 - \gamma) pq} + p - q}  \right)}
  {\frac{p + pq - q}{p + pq - q^2} + \frac{(1 - \gamma) pq}{p + pq - q^2} 
  (A + \lambda) \left( \frac{1}{[A + 1]} - \frac{1}{A+\lambda} \right)}
  $$
\end{theorem*}

  The polynomial $1 - A X^{A - 1} + A X^{A + 1} - X^{2 A}$ vanish at $X = 1$, as well as its two first
  derivatives, hence $P_A (X)$ is a polynomial in $\mathbb{Z} [X]$. 
  Making $A=1$ in Theorem \ref{thm_main} we get Theorem 1 of \cite{GPM2018-2} as a particular case.
  Indeed, we have ${\text{TSM}}_1 = {\text{LSM}}$, i.e. $1$-Trail-Stubborn Mining and Lead-Stubborn Mining 
  strategies are the same. Before proving Theorem \ref{thm_main} we need to study a classical refinement 
  of the Gambler's Ruin Problem: The hiker problem.

\section{The hiker problem.} \label{sec_hicker}

We consider a hiker on $[ 0, M ]$ with $M \in \mathbb{N}$, $M\geq 2$. 
His position is denoted by the random process $(\mathbf{X}_n)_{n \in
\mathbb{N}}$. The transition probability from $i$ to $j$ are
$$ 
P (i, j) =\mathbb{P} [\mathbf{X}_{n + 1} = j| \mathbf{X}_n = i] =
   p\, {\bf{1}}_{i = j - 1} + q\, {\bf{1}}_{i = j + 1} 
$$
for $(i, j) \in [ 1, M - 1 ] \times [ 0, M
]$. It is independent of $n \in \mathbb{N}$. We make the assumption
that $0$ and $M$ are absorbing boundaries: $P (0, 0) = P (M, M) = 1$. 
The problem is studied in \cite{F} where it is proved that with probability $1$ the hiker exits 
$[ 1, M - 1 ]$. We need more precise information.

\begin{definition}
  \label{nm}For $k \in \{ 0, M \}$ and $m \in [ 0, M ]$, let
  $\nu_{m, k} \in \mathbb{N} \cup \{ \infty \}$ be the stopping time defined
  by 
  $$
  \nu_{m, k} = \inf \{ n ; \mathbf{X}_n = k| \mathbf{X}_0 = m \}
  $$ 
  
\end{definition}

We denote $\nu_m = \nu_{m, 0} \wedge \nu_{m, M}$, i.e. 
$\nu_m$ is the stopping time for exiting $[ 1, M -1 ]$ 
starting from $m$. From \cite{F}, we have
$\nu_m <+\infty$ almost surely. The condition $\nu_m = \nu_{m, 0}$
is equivalent to the realization of the event ``the hiker exits $[ 1, M - 1 ]$ 
at $1$''.

\begin{theorem} \label{en}
We have:
  \begin{align*}
  \mathbb{E} [\nu_m] & = \frac{M}{p - q} \cdot \left( \frac{1 -
    \lambda^m}{1 - \lambda^M} - \frac{m}{M} \right)\\
    \mathbb{P} [\nu_m = \nu_{m, 0}] & = \frac{\lambda^m - \lambda^M}{1 -
    \lambda^M}\\
    \mathbb{E} [\nu_m | \nu_m = \nu_{m, 0}] & = \frac{m \lambda^m - (2 M - m) 
    \lambda^M + (2 M - m) \lambda^{M + m} - m \lambda^{2 M}}{p (1 -
    \lambda)  (\lambda^m - \lambda^M)  (1 - \lambda^M)}
  \end{align*}
\end{theorem}

The first two equations are well known classical results that can be found in \cite{F} p. $314$ and p.$317$. 
The last equation is from \cite{St} and is not so classical and, to be self-contained, we give another proof in Appendix \ref{rw}.

\begin{corollary}
  We have $\underset{M \rightarrow \infty}{\lim} \mathbb{E} [\nu_m | \nu_m =
  \nu_{m, 0}] = \frac{m}{p - q}$.
\end{corollary}

\begin{definition}
  We denote by $\mathcal{L} (n)$ (resp. $\mathcal{R} (n)$) the number of steps
  to the left (resp. right) realized by the hiker between $t = 0$ and $t = n$.
\end{definition}

In other terms, $\mathcal{L} (0) =\mathcal{R} (0) = 0$ and for $n \leq
\nu_m,$
\begin{align*}
  \mathcal{L} (n) & =  \mathcal{L} (n - 1) +{\bf{1}}_{\mathbf{X} (n) =
  \mathbf{X} (n - 1) - 1}\\
  \mathcal{R} (n) & =  \mathcal{R} (n - 1) +{\bf{1}}_{\mathbf{X} (n) =
  \mathbf{X} (n - 1) + 1}
\end{align*}
Note that $\mathcal{L} (n) +\mathcal{R} (n) = n$ for $n \leq \nu_m$.

\begin{corollary} \label{lrs}
We have 
\begin{align*}
  \mathbb{E} [\mathcal{L} (\nu_m) | \nu_m = \nu_{m, 0}] & =  \frac{m}{2} + 
  \frac{m \lambda^m - (2 M - m) \lambda^M + (2 M - m) \lambda^{M + m} - m \lambda^{2 M}}
  {2 p (1 - \lambda)  (\lambda^m - \lambda^M)  (1 - \lambda^M)}\\
  \mathbb{E} [\mathcal{R} (\nu_m) | \nu_m = \nu_{m, M}] & =  \frac{M - m}{2} + 
  \frac{m \lambda^m - (2 M - m) \lambda^M + (2 M - m) \lambda^{M + m} - m \lambda^{2 M}}
  {2 p (1 - \lambda)  (\lambda^m - \lambda^M)  (1 - \lambda^M)}
\end{align*}

\end{corollary}

\begin{proof}
  If the hiker exits $[ 1, M - 1 ]$ at 1 (resp. $M-1$),
  then $\mathcal{L} (\nu_m) =\mathcal{R} (\nu_m) + m$ 
  (resp. $\mathcal{R} (\nu_m) =\mathcal{L} (\nu_m) + M - m$). So we have

\begin{align*}
  \mathbb{E} [\mathcal{L} (\nu_m) | \nu_m = \nu_{m, 0}] & = 
  \frac{m}{2} + \frac{\mathbb{E} [\nu_m | \nu_m = \nu_{m, 0}]}{2}\\
  \mathbb{E} [\mathcal{L} (\nu_m) | \nu_m = \nu_{m, M}] & = 
  \frac{M-m}{2} + \frac{\mathbb{E} [\nu_m | \nu_m = \nu_{m, 0}]}{2}
\end{align*}
 and the result follows from Theorem \ref{en}.
\end{proof}

In particular, for $m=2$ this gives:
\begin{equation}\label{elnu20}
  \mathbb{E} [\mathcal{L} (\nu_2) | \nu_2 = \nu_{2, 0}] = 
  1 + \frac{1}{p} . \frac{1 - (M-1) \lambda^{M-2} + (M - 1) \lambda^{M} - \lambda^{2 M - 2}}
  {(1 - \lambda)  (1 - \lambda^{M-2})  (1 - \lambda^M)}
\end{equation}

\section{Expected duration of the Trail-Stubborn Mining Strategy}

\subsection{Notations and previous results.}
We set, $\alpha = \frac{p}{\tau_0}, \alpha' = \frac{q}{\tau_0}$ and \
$\lambda = \frac{q}{p} < 1$. We note $N$ and $N'$ the two
independent Poisson processes with parameters $\alpha$ and $\alpha'$
representing the number of blocks validated by the honest miners and the rogue
miner. We denote by $T_1, T_2, \ldots$ (resp. $T'_1, T'_2, \ldots$) 
the inter-block validation time for the honest miners (resp. attackers).

We use some notations from \cite{GPM2018}. In particular, for the stopping times:
\begin{equation}
  \tau = \inf \{ t \geq T_1 ;\, N (t) = N' (t) +{\bf{1}}_{T_1 < T'_1} \}
\end{equation}
and
\begin{equation}
  \tau_{{LSM}} = \tau + (T_{N (\tau) + 1} \wedge T'_{N (\tau) + 1}) \cdot {\bf{1}}_{T'_1 \leq T_1}
\end{equation}
We proved in \cite{GPM2018} that 

\begin{equation}\label{etlsm}
  \mathbb{E} [\tau] = \frac{p}{p - q} \tau_0,\, 
  \mathbb{E}[\tau_{LSM}] = \left( \frac{p}{p - q} + q \right) \tau_0,\, \text{and}\quad
  \mathbb{E} [N' (\tau)] = \alpha' \mathbb{E} [\tau] = \frac{pq}{p - q}
\end{equation}
More precisely, for $n\geq 1$, we have
\begin{equation}
  \mathbb{P} [N' (\tau) = n]  =  C_{n - 1}  (pq)^n 
\end{equation}
where $C_n = \frac{1}{n + 1}  \binom{2n}{n}$ denotes the $n$-th Catalan number, whose 
generating series is $C (x) = \frac{1-\sqrt{1-4x}}{2x}$.

At the end of an attack cycle, the revenue of a miner following the Lead
Stubborn Mining strategy is denoted by $R (\tau_{{LSM}})$.

\subsection{Description of the $A$-Trail-Stubborn mining strategy} At the beginning of an attack cycle both, 
the rogue miner and the honest miners start mining on top of the same common block. Then,
either the first block is discovered by the honest miners, and the attack cycle 
ends, or the attacker is the first one validating a block. Then, he keeps mining 
secretly until he is being caught up by the honest miners. During this period,
each time the honest miners publish a new block, 
the rogue miner broadcasts the part of his fork sharing the same height.
Once he has been caught up, there is a ``decisive competition" to decide which fork prevails. 
In this competition, the rogue miner does not withold his block. 
There are two cases depending on who the winner is.
Either the new block is mined on top of a block validated by the rogue miner 
(by himself or by a fraction $\gamma$ of the honest miners) and then the attack cycle ends immediatly.
Otherwise, the attacker has a fork which is one block behind the official blockchain. We call this event $\mathbf{\Sigma}$.
His delay is defined as the difference between the height of the official blockchain and his fork.
Then, he keeps mining until his delay exceeds a fixed threeshold $A\geq 1$, 
or ends up leading the official blockchain by one block. Then, in both cases, the cycle attack ends.
The trailing mining strategy is a repetition of these attack cycles.

\subsection{Stopping time}We denote by $\xi$ the stopping time of an attack
cycle corresponding to the Trail-Stubborn mining strategy (TSM). At the end of
an attack cycle the revenue of a rogue miner following TSM is denoted by $R(\xi)$.
Note that when $T'_1 < T_1$, there is a ``decisive round" which starts at $\tau$ and ends at $\tau_{LSM}$. 
So, $\tau_{LSM}\leq \xi$ and from $0$ to $\tau_{LSM}$ the two strategies $\tau_{LSM}$ and $\xi$ are the same.

\subsection{Blocks of the rogue miner in the official blockchain}For $t
\geqslant 0$, we denote by $Z (t)$ the number of blocks mined by a miner
following the Trail-Stubborn Mining strategy at $t$-time and present in the official blockchain. 
We have that $t \mapsto Z (t)$ is non-decreasing. Before $t\leq \tau$, the two strategies 
``Trail-Stubborn Mining Strategy'' and``Lead Stubborn Mining Strategy'' are the same. 
So, by \cite{GPM2018-2} we know that
\begin{equation}
  \mathbb{E} [Z (\tau) |N' (\tau) = n]  =  n - \frac{1 - (1 -
  \gamma)^n}{\gamma}  \label{rten}
\end{equation}

\begin{lemma}\label{xitaulsm}
The following conditions are equivalent to $\mathbf{\Sigma}$:
\begin{enumerate}[label=(\roman*)\quad]
\item $\tau_{LSM} < \xi$;
\item $R(\tau_{LSM}) < N'(\tau_{LSM}) b$;
\item $( T'_1 < T_1)\wedge (T_{N (\tau) + 1} < T'_{N (\tau) + 1})\ \wedge$ 
(the $\left(N (\tau) + 1\right)$-th honest
block is found on top of a block mined by a honest miner).
\end{enumerate}
If one of these conditions is satisfied then $N'(\tau_{LSM}) = N'(\tau)$, $N(\tau_{LSM}) = N(\tau) + 1$ and $Z(\tau_{LSM}) = Z(\tau)$.
We have that $\mathbf{\Sigma}$ is $\tau_{LSM}$-measurable and $\mathbb{P}[\mathbf{\Sigma}]=(1-\gamma) p q$.
\end{lemma}

\begin{proof}
If (i) holds then $T'_1 < T_1$ otherwise $\tau_{LSM} = \tau = \xi = T_1$. 
Moreover, at $\tau_{LSM}$, at least one block mined by
the rogue miner has not been recognized by the official blockchain 
(otherwise, the cycle ends at $\tau_{LSM}$ and $\xi = \tau_{LSM}$). 
So, $R(\tau_{LSM}) < N'(\tau_{LSM}) b$.
If (ii) is true then $0<N'(\tau_{LSM})$. 
So, the miner has at least mined a block during the attack cycle. 
So, $T'_1 < T_1$ (otherwise as before $\tau_{LSM} = \tau = \xi = T_1$ and $N'(\tau_{LSM})=0$)). 
Moreover, the rogue miner has lost the ``decisive competition". Otherwise, 
we have $R(\tau_{LSM}) = N'(\tau_{LSM}) b$.
Also, by the same argument, the block found by the honest miners during this round 
cannot have been validated on top of a block mined by the rogue miner. 
So, we get (iii). Finally, if (iii) holds, then the rogue miner has lost the ``decisive competition". 
Hence, $\tau_{LSM} < \xi$ by definition of an attack cycle and so (i), (ii) and (iii) are equivalent.
Also, if one of these conditions is satisfied, then the rogue miner did not mine a block during the 
period $[\tau,\tau_{LSM}]$ whereas the honest miner has found exactly one.
So, $N'(\tau_{LSM}) = N'(\tau)$ and $N(\tau_{LSM}) = N(\tau) + 1$. 
Moreover the block found by the honest miners has been found on top of an honest block by (iii). 
So we have, $Z(\tau_{LSM}) = Z(\tau)$. By (ii), $\mathbf{\Sigma}$ is $\tau_{LSM}$-measurable. Moreover,
the condition $\{ T'_1 < T_1 \}$ occurs with probability $q$ and the two last conditions of (iii) occur with probability $(1 - \gamma) p$.
Therefore, we have $\mathbb{P}[\mathbf{\Sigma}]=(1-\gamma) p q$.
\end{proof}

\subsection{Trail-mining.}

  We consider that after a possible second phase of the attack cycle (after $\tau_{LSM}$), the
  rogue miner following ${\text{TSM}}_A$ will give up if his delay exceeds $A$ with $A \geqslant 1$. 
  Note that at the beginning of this second phase, the delay of the miner is $1$. So, ${\text{TSM}}_1 = {\text{LSM}}$.
  Note also that in order to win, it is not enough for the miner to catch-up the official blockchain. 
  He needs to lead it by $1$ block. So, his delay is in between $-1$ and $A$. Therefore, 
  he behaves as the hiker studied in Section \ref{sec_hicker} with delay $\mathbf{X}_{n} - 1$ and $M=A+1$.

\begin{proposition}\label{marko}
  We have 
  $$
  \xi = \tau_{LSM} + \sigma . {\bf 1}_{\mathbf{\Sigma}}
  $$ 
  where $\sigma$ is the stopping time defined by
  $$
  \sigma = \inf \{ t \in \mathbb{R}_+^{\ast} ; (\tilde{N}' (t) =
  \tilde{N} (t) + 2) \vee (\tilde{N} (t) = \tilde{N}' (t) + A - 1) \}
  $$ 
  with $\tilde{N} (t) = N (t + \tau_{LSM}) - N (\tau_{LSM})$ and $\tilde{N}' (t) = N'
  (t + \tau_{{LSM}}) - N' (\tau_{LSM})$.
  
  In particular, we have that $\tilde{N}$ and
  $\tilde{N}'$ are two independent Poisson processes with respective parameters $\alpha$
  and $\alpha'$, and $\sigma$ is independent with $\mathbf{\Sigma}$.
\end{proposition}

\begin{proof}
  The stopping time of the Trail-Stubborn Mining Strategy is the same as the
  stopping time of the Lead Stubborn Mining strategy studied in \cite{GPM2018-2} 
  except when the stubborn
  miner has been first mining a block, then has been caught-up by the honest
  miners, and at has lost the final ``competition'' (when a fraction $(1 - \gamma) p$ of the honest miners finds a new
  block on top of a honest block).
  We have called $\mathbf{\Sigma}$ this event. If it occurs, then
  the stubborn miner keeps on mining until he catches up the honest miners or
  his delay becomes too big. In this case, he  needs to catch-up the
  honest miners and also lead the official blockchain by $1$. The start time
  of this possible second round is $\tau_{LSM}$ with $N'
  (\tau_{LSM}) = N (\tau_{LSM}) - 1$ and the miner will stop at
  $\tau_{LSM} + t$ with $N' (t + \tau_{LSM}) = N (t +
  \tau_{{LSM}}) + 1$ or $N (t + \tau_{{LSM}}) = N' (t + \tau_{{LSM}}) + A$. 
  The first equality is equivalent to 
  $N' (t + \tau_{{LSM}}) - N' (\tau_{{LSM}}) = N (t + \tau_{{LSM}}) - N
  (\tau_{{LSM}}) + 2$ and the second is equivalent to $N (t +
  \tau_{{LSM}}) - N (\tau_{{LSM}}) = N' (t + \tau_{{LSM}}) - N'
  (\tau_{{LSM}}) + A - 1$. 
  Moreover, by the strong Markov property $\sigma$ is independent with $\tau_{LSM}$. 
  So, by Lemma \ref{xitaulsm}, it is also independent with $\mathbf{\Sigma}$.
\end{proof}

Note that the condition $(\tilde{N}' (\sigma) = \tilde{N} (\sigma) + 2) \vee (\tilde{N} (\sigma) = \tilde{N}' (\sigma) + A - 1)$
is equivalent to $X (\sigma) \in \{ 0, A+1 \}$ with $X (t) = N (t) - N' (t) + 2$. So we have that 
the miner is a  hiker on $[ 0, M]$ as studied in section \ref{sec_hicker} starting from $\mathbf{X}_0 = 2$ with $M=A+1$.
In Appendix A we prove the following Proposition:

\begin{proposition} \label{disnu2}
We have 
  $$\mathbb{E} [\sigma] = \frac{A+1}{p - q}  \left( \frac{1 - \lambda^2}{1 -
  \lambda^{A+1}} - \frac{2}{A+1} \right) \tau_0
  $$
\end{proposition}

\begin{proposition}\label{exiTSM}
  We have 
  $$
  \frac{\mathbb{E} [\xi]}{\tau_0} = \frac{p}{p - q} + q + (A+1) \cdot
  \frac{(1 - \gamma) pq}{p - q} \cdot \left( \frac{1 - \lambda^2}{1 -
  \lambda^{A+1}} - \frac{2}{A+1} \right)
  $$
\end{proposition}

\begin{proof}
  By Proposition \ref{marko}, we have
\begin{eqnarray*}
    \mathbb{E} [\xi] & = & \mathbb{E} [\tau_{LSM}] + \mathbb{P} [\mathbf{\Sigma}] \cdot \mathbb{E} [\sigma]
  \end{eqnarray*}
  So, we get the result using (\ref{etlsm}), Lemma \ref{xitaulsm} and Proposition \ref{disnu2}.
\end{proof}

\section{Revenue ratio of the Trail-Stubborn Mining Strategy}

\begin{proposition} We have:
  $$ 
  R (\xi) = R (\tau_{{LSM}}) \cdot {\bf{1}}_{\xi =
     \tau_{{LSM}}} + (N' (\tau) +\mathcal{L} (\nu_2)) b \cdot
     {\bf{1}}_{(\xi > \tau_{{LSM}}) \wedge (\nu_2 = \nu_{2, 0})} + Z
     (\tau) b \cdot {\bf{1}}_{(\xi > \tau_{{LSM}}) \wedge (\nu_2 =
     \nu_{2, A+1})}
  $$
\end{proposition}
  In this Proposition $\mathcal{L} (\nu_2)$ is the number of blocks validated by the rogue
  miner during the second phase of the strategy. The event $(\xi >
  \tau_{{LSM}}) \wedge (\nu_2 = \nu_{2, 0})$ (resp. $(\xi >
  \tau_{{LSM}}) \wedge (\nu_2 = \nu_{2, A+1})$) means that the cycle is
  made of two distinct phases: in the first one the rogue miner looses the first phase of
  the attack, and in the second one he wins (resp. looses) the second phase.

\begin{proof}
  If $R (\tau_{{LSM}}) < N' (\tau_{{LSM}}) b$, then the
  miner tries to catch-up the official blockchain. He is in the position of a
  hiker starting from $\mathbf{X}_0 = 2$ and winning when $\nu_2 = \nu_{2,
  0}$. Each move to the left (towards $0$) corresponds to a new
  block mined by the stubborn miner. So, if he succeeds (case $\nu_2 = \nu_{2,
  0}$), then he earns a reward $(N' (\tau_{LSM}) +\mathcal{L} (\nu_2)) b$. If he fails (case
  $\nu_2 = \nu_{2, A+1}$) then he earns only $Z (\tau_{LSM}) b$ and the attack cycle
  ends. Otherwise, the strategy ends at $\tau_{LSM}$ and $R(\xi) = R (\tau_{LSM})$. 
  The result then follows from Lemma \ref{xitaulsm}.
\end{proof}

Now we compute the expected revenue of the $A$-Trail-Stubborn Mining Strategy in an attack cycle.

\begin{proposition} \label{rxi}
We have
  \begin{align*}
    &\frac{\mathbb{E} [R (\xi)]}{b}  = \left( \frac{p + pq - q^2}{p - q} \right) q + (1 - \gamma) pq \, \cdot \\
    &\cdot \left[ \left( 1 +
    \frac{1}{p} \cdot \frac{1 - A \lambda^{A - 1} + A \lambda^{A+1} -
    \lambda^{2 A}}{(1 - \lambda)  (1 - \lambda^{A - 1})  (1 - \lambda^{A+1})}
    \right)  \frac{\lambda^2 - \lambda^{A+1}}{1 - \lambda^{A+1}} - \frac{2 p}{\sqrt{1
    - 4 (1 - \gamma) pq} + p - q}  \frac{1 - \lambda^2}{1 - \lambda^{A+1}}
    \right]
  \end{align*}
\end{proposition}

\begin{proof}
  Consider the events, for $n \in \mathbb{N}$, $E_n = \{ N' (\tau) = n \}$, $F = \{ R
  (\tau_{{LSM}}) = N' (\tau_{{LSM}}) b \}$ and $G = \{ \nu_2 =
  \nu_{2, 0} \}$. From \cite{GPM2018} and \cite{GPM2018-2} we have
  \begin{align*}
    \mathbb{P} [E_n] & =  p \, {\bf{1}}_{n = 0} + (pq)^n C_{n - 1}
    {\bf{1}}_{n > 0}  \label{olda}\\
    \mathbb{P} [G] & = \frac{\lambda^2 - \lambda^{A+1}}{1 - \lambda^{A+1}} 
  \end{align*}
  and for $n > 0$,
  \begin{equation*}
    \mathbb{P} [E_n \cap F]  =  \mathbb{P} [E_n]  (q + \gamma p) 
  \end{equation*}
  Note also that
  \begin{itemize}
    \item If $E_0$ occurs then $R (\xi) = 0$.
    
    \item If $E_n \cap F$ occurs (with $n > 0$) then $R (\xi) = (n + 1) b$ with
    probability $\frac{q}{q + \gamma p}$ and $R (\xi) = nb$ with probability
    $\frac{\gamma p}{q + \gamma p}$.
    
    \item If $E_n \cap \bar{F} \cap \bar{G}$ occurs then $R (\xi) = Z (\tau) b$.
    
    \item If $E_n \cap \bar{F} \cap G$ occurs then $R (\xi) = (n +\mathcal{L}
    (\nu_2)) b$.
  \end{itemize}
  So, by conditioning on $E_0, E_n \cap F, E_n \cap \bar{F} \cap G, E_n \cap
  \bar{F} \cap \bar{G}$ and using (\ref{rten}) and Corollary \ref{lrs} 
  together with $\sum_{n \geqslant 0} \mathbb{P} [E_n] = 1$ we
  have:
  \begin{align*}
    \frac{\mathbb{E} [R (\xi)]}{b}  &=  0 \cdot \mathbb{P} [E_0] + \sum_{n
    > 0} ((n + 1) q + n \gamma p) \, \mathbb{P} [E_n]\\
    & \ +  \sum_{n > 0} \left( n - \frac{1 - (1 - \gamma)^n}{\gamma} \right) 
    (1 - \gamma) p \frac{1 - \lambda^2}{1 - \lambda^{A+1}} \, \mathbb{P} [E_n]\\
    & \ +  \sum_{n > 0} \left( n + \mathbb{E} [\mathcal{L} (\nu_2) | \nu_2 = \nu_{2, 0}] \right)
    (1 - \gamma) p \frac{\lambda^2 - \lambda^{A+1}}{1 - \lambda^{A+1}} \, \mathbb{P} [E_n]\\
    & =  \mathbb{E} [N' (\tau)] + \frac{(1 - \gamma) p}{\gamma}  \frac{1 -
    \lambda^2}{1 - \lambda^{A+1}}  (1 - \gamma) pq\, C ((1 - \gamma) pq)\\
    & \ +  \left( q - \frac{(1 - \gamma) p}{\gamma}  \frac{1 - \lambda^2}{1 -
    \lambda^{A+1}} + (1 - \gamma) p \mathbb{E} [\mathcal{L} (\nu_2) | \nu_2 = \nu_{2, 0}]
    \frac{\lambda^2 - \lambda^{A+1}}{1 - \lambda^{A+1}} \right)  (1 -\mathbb{P}
    [E_0])\\
    & =  \left( \frac{p}{p - q} + q \right) q - (1 - \gamma) pq \frac{1 -
    \lambda^2}{1 - \lambda^{A+1}}  \frac{[1 - (1 - \gamma) pC ((1 - \gamma)
    pq)]}{\gamma}\\
    & \ +  (1 - \gamma) pq \mathbb{E} [\mathcal{L} (\nu_2) | \nu_2 = \nu_{2, 0}] 
    \frac{\lambda^2 - \lambda^{A+1}}{1 - \lambda^{A+1}}
  \end{align*}
  Moreover for $q > 0$, we have
  \begin{align*}
    1 - (1 - \gamma) p\, C ((1 - \gamma) pq) & =  1 - \frac{1 - \sqrt{1 - 4 (1 -
    \gamma) pq}}{2 q}\\
    & =  \frac{\sqrt{1 - 4 (1 - \gamma) pq} - (p - q)}{2 q}\\
    & =  \frac{1 - 4 pq + 4 \gamma pq - (p^2 - 2 pq + q^2)}{2 q \left[
    \sqrt{1 - 4 (1 - \gamma) pq} + p - q \right]}\\
    & =  \frac{2 p \gamma}{\sqrt{1 - 4 (1 - \gamma) pq} + p - q}
  \end{align*}
  and we get the result using (\ref{elnu20}).
\end{proof}

Proposition \ref{exiTSM} and Proposition \ref{rxi} give the revenue ratio of the strategy and the first part of Theorem \ref{thm_main}.

\section{Difficulty adjustment}\label{sec_diff_adjust}

\begin{proposition}
  We have
  $$
  \mathbb{E} [N (\xi) \vee N' (\xi)] = \frac{pq + p - q}{p - q} + \frac{(1
     - \gamma) pq}{p - q}  \left( (A p + q)  \frac{1 - \lambda^2}{1 -
     \lambda^{A+1}} - 1 \right) 
  $$
\end{proposition}

\begin{proof}
  We keep the same notations as in the proof of Proposition \ref{rxi}. Note
  that
  \begin{itemize}
    \item If $E_0$ occurs, then $N (\xi) \vee N' (\xi) = 1$.
    
    \item If $E_n \cap F$ occurs ($n > 0$), then $N (\xi) \vee N' (\xi) = n + 1$.

    \item If $E_n \cap \bar{F} \cap G$ occurs ($n > 0$), then $N (\xi) \vee N'
    (\xi) = n +\mathcal{L} (\nu_2)$.
    
    \item If $E_n \cap \bar{F} \cap \bar{G}$ occurs ($n > 0$), then $N (\xi) \vee N' (\xi) = N (\xi) = n+1 +\mathcal{R}(\nu_2)$.
  \end{itemize}
  So, by conditioning as before and using Corollary \ref{lrs}, we get
  \begin{align*}
    \mathbb{E} [N (\xi) \vee N' (\xi)] & =  1 \cdot \mathbb{P} [E_0] +
    \sum_{n > 0}  (n + 1) (q + \gamma p) \, \mathbb{P} [E_n]\\
    & \ +  \sum_{n > 0}  \left( n + 1 + \frac{1}{2} \mathbb{E}[\nu_2 | \nu_2
    = \nu_{2, 0}] \right)  (1 - \gamma) p \, \mathbb{P} [E_n] \, \mathbb{P} [\nu_2
    = \nu_{2, 0}]\\
    & \ +  \sum_{n > 0}  \left( n + 1 + \frac{A+1}{2} - 1 + \frac{1}{2}
    \mathbb{E} [\nu_2 | \nu_2 = \nu_{2, A+1}] \right)  (1 - \gamma)
    p \, \mathbb{P} [E_n] \, \mathbb{P} [\nu_2 = \nu_{2, A+1}]\\
    & =  \mathbb{P} [E_0] + \sum_{n > 0}  (n + 1) \mathbb{P} [E_n] +
    \sum_{n > 0}  \left( \frac{A+1}{2} - 1 \right)  (1 - \gamma) p \, \mathbb{P}
    [E_n] \, \mathbb{P} [\nu_2 = \nu_{2, A+1}]\\
    & \ +  \sum_{n > 0}  (1 - \gamma) p\, \mathbb{P} [E_n]  \, \frac{\mathbb{E}
    [\nu_2]}{2}\\
    & =  \mathbb{E} [N' (\tau)] + 1 + (1 - \gamma) pq \left( \left(
    \frac{A+1}{2} - 1 \right) \mathbb{P} [\nu_2 = \nu_{2, A+1}] +
    \frac{\mathbb{E} [\nu_2]}{2} \right)\\
    & =  \frac{pq + p - q}{p - q} + \frac{(1 - \gamma) pq}{p - q}  ((A
    p + q) \, \mathbb{P} [\nu_2 = \nu_{2, A+1}] - 1)
  \end{align*}
\end{proof}

\begin{theorem}
  The parameter $\delta$ updating the difficulty of the $A$-trail stubborn mining strategy is given by
  \begin{equation*}
    \delta = \frac{\frac{p + pq - q^2}{p - q} + (A+1) \cdot \frac{(1 - \gamma)
    pq}{p - q} \cdot \left( \frac{1 - \lambda^2}{1 - \lambda^{A+1}} - \frac{2}{A+1}
    \right)}{\frac{pq + p - q}{p - q} + \frac{(1 - \gamma) pq}{p - q}  ((A p + q) \frac{1 - \lambda^2}{1 - \lambda^{A+1}} - 1)}
  \end{equation*}
\end{theorem}

\begin{proof}
  From  \cite{GPM2018} we have 
  that $\delta = \frac{\mathbb{E} [\xi]}{\mathbb{E} [N (\xi) \vee N' (\xi)]}
  \cdot \frac{1}{\tau_0}$.
\end{proof}

\subsection{Observations}
We denote by ${\tilde q}_{A}$ the long-term apparent hashrate 
of the $A$-Trail-Stubborn-mining strategy.
As we have already observed, the lead-stubborn mining strategy LSM 
is a particular case of the $A$-Trail-Stubborn-mining strategy with $A=1$
(in tis case, there is no possible second phase of the attack 
after $\tau_{LSM}$). We note that Theorem \ref{thm_main} yields one of the results of \cite{GPM2018-2}: if we choose $A=1$ 
in Theorem \ref{thm_main}, we get Theorem 1 of \cite{GPM2018-2}.
In Figure 1 below, we compare ${\tilde q}_{{\text{LSM}}}$ 
(the long-term apparent hashrate of the strategy LSM) with ${\text{Max}}\{ {\tilde q}_{A}\, ;\, A\geq 2\}$. 
Depending on $(q,\gamma)$, this shows when a second phase 
of the attack increases the efficiency of the strategy LSM. 
In general, when $\gamma$ is small, TSM is an amelioration of LSM.

\begin{figure}[h]\label{lsmtsm-fig}
  \resizebox{400pt}{300pt}{\includegraphics{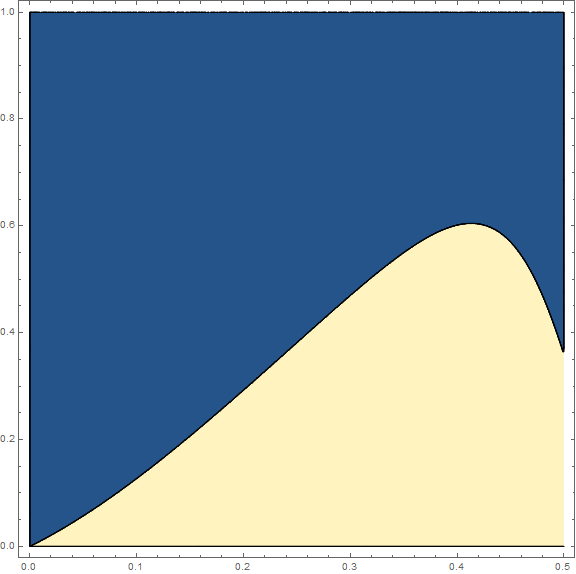}}
  \caption{LSM vs $A$-Trail-Stubborn Mining strategy for $A \geq 2$}
\end{figure}

For $\gamma$ greater than $20\%$, TSM with $A=2$ dominates 
all other trailing strategies whatever $q$ is. See Figure 2.

Also, if $\gamma = 0$, then
  \begin{equation*}
    \widetilde{q_{}}_{{TSM}}  =  \frac{\left( \frac{p + pq - q^2}{p -
    q} \right) q + pq \left[ \left( 1 + \frac{1}{p} \cdot \frac{1 - A
    \lambda^{A - 1} + A \lambda^{A+1} - \lambda^{2 A}}{(1 - \lambda)  (1
    - \lambda^{A - 1})  (1 - \lambda^{A+1})} \right)  \frac{\lambda^2 -
    \lambda^{A+1}}{1 - \lambda^{A+1}} - \frac{p}{p - q}  \frac{1 - \lambda^2}{1 -
    \lambda^{A+1}} \right]}{1 + \frac{pq}{p - q}  \left( (A p + q)  \frac{1
    - \lambda^2}{1 - \lambda^{A+1}} \right)}
  \end{equation*}
  On the other hand, when $\gamma = 0$, the apparent hashrate after a
  difficulty adjustment for the selfish mining strategy is
  $\widetilde{q_{}}_{{SM}} = \frac{pq^2 + (p - q)  (q + pq^2 - p^2
  q)}{p^2 q + p - q}$. It turns out that for $\gamma = 0$,
  \[ \underset{q \rightarrow \frac{1}{2}}{\lim}  \widetilde{q_{}}_{{TSM}}
     = 1 - \frac{1}{A+1} < 1 = \underset{q \rightarrow \frac{1}{2}}{\lim} 
     \widetilde{q_{}}_{{SM}} \]
  Hence, when $q \rightarrow \frac{1}{2}$ and $\gamma \ll 1$, SM dominates all TSM strategies.

\begin{figure}[h]
  \resizebox{400pt}{300pt}{\includegraphics{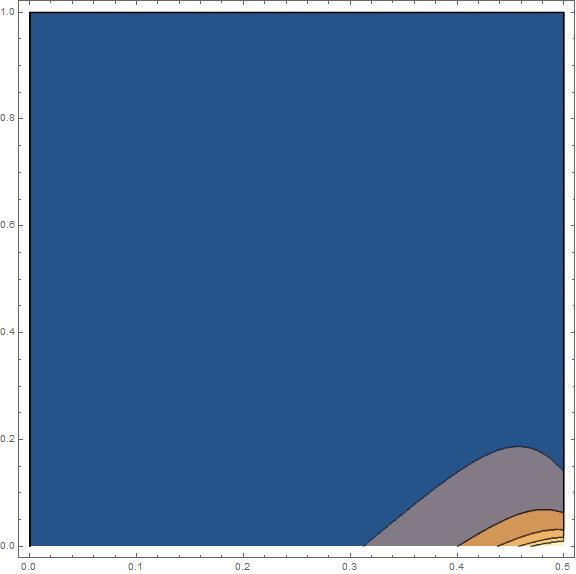}}
  \caption{$A$-Trail-Stubborn Mining strategy for $A = 2, 3, 4, 5, 6, 7$}
\end{figure}

\pagebreak

\section{Mixed strategies.}

\subsection{Weight of a mining strategy} We consider a miner mining according
to a strategy $\tau$.

\begin{definition}
  The weight of a mining strategy is the average number of official blocks
  mined during an attack cycle. It is denoted by the greek letter $\mu$.
\end{definition}

Note that if the strategy leads to a difficulty adjustment $D$, then we have:
$\mu = \frac{\mathbb{E} [\tau]}{\tau_0 D}$.

\subsection{Apparent hashrate of a mixed strategy}We consider now a miner
implementing a mixed strategy. He starts mining according to strategy 1, 
then at the end of an attack cycle, he decides to follow another strategy 2, and so on 
until he comes back to
strategy 1 after implementing $n$ of different strategies. 
Thus the attack cycle is a given pattern of attack cycles of different
strategies.

We denote by $\Gamma_1, \Gamma_2, \ldots,  R_1, R_2,\ldots \tau_{1},\tau_2,\ldots   
\tilde{\Gamma}_1, \tilde \Gamma_2, \ldots D_1, D_2\ldots \mu_1,\mu_2,\ldots$
the revenue ratio, revenue, duration time, long-term apparent hashrate, difficulty
adjustment and weight over an attack cycle of strategy $1, 2, \ldots$. 
We denote by $\Gamma, R, \tau, \tilde{\Gamma}, D$ and $\mu$, the revenue
ratio, revenue, duration time, long-term apparent hashrate, difficulty
adjustment and weight after an attack cycle of the mixed strategy.

\begin{theorem} \label{mixb}We have that $(D, \tilde{\Gamma})$ is barycenter of $(D_1,
  \tilde{\Gamma}_1), (D_2, \tilde{\Gamma}_2), \ldots $ weighted by $\mu_1, \mu_2,\ldots$.
\end{theorem}

\begin{proof}
  The number
  $\mu$ of official blocks mined after a whole attack cycle of the mixed
  strategy is $\mu = \mu_1 + \mu_2+\ldots +\mu_n$. Moreover, $\mathbb{E} [\tau]
  =\mathbb{E} [\tau_1] +\mathbb{E} [\tau_2]+\ldots +\mathbb{E} [\tau_n]$. Therefore,
  \begin{equation*}
    D  =\sum_{i=0}^n  \frac{\frac{\mathbb{E} [\tau_i]}{\tau_0 D_i}}{\mu} D_i = \sum_{i=0}^n  \frac{\mu_i}{\mu} D_i
  \end{equation*}
  Similarly, we have
  \begin{align*}
    \tilde{\Gamma} & =  \Gamma D =  \frac{\mathbb{E} [R]}{\mathbb{E} [\tau]} \, D 
    = \frac{\sum_{i=1}^n \mathbb{E} [R_i]}{\mathbb{E} [\tau]} D
    = \frac{1}{\mu}\sum_{i=1}^n \mathbb{E} [R_i]\\
    & = \frac{1}{\mu}\sum_{i=1}^n \Gamma_i\mathbb{E} [\tau_i]
    = \frac{1}{\mu}\sum_{i=1}^n \tilde{\Gamma}_i  \frac{\mathbb{E} [\tau_i]}{D_i}
    = \frac{1}{\mu}\sum_{i=1}^n \mu_i  \tilde{\Gamma}_i
  \end{align*}
  
\end{proof}

\begin{corollary}
  We have $\tilde{\Gamma} \leq \max (\tilde{\Gamma}_1,
  \tilde{\Gamma}_2, \ldots ,\Gamma_n)$ with equality if and only if the strategy is not mixed.
\end{corollary}

Therefore, there is no advantage in implementing mixed strategies.

\section{Comparison with other strategies}

We compare Trailing-Stubborn Mining strategies with $A = 2, 3, 4$ and other strategies 
HM, SM, LSM and EFSM studied in \cite{GPM2018-2}. We observe 
that LSM is the dominant strategy in a very thin region between SM, EFSM and TSM2. 
Below to the right (but for
$\gamma$ not too small), the dominant strategy is TSM3. The strategy TSM4 is
dominant in a very little domain with $q \approx 0.5$ and $\gamma \approx 5$\%.
For $\gamma$ less than 5\% and large $q$ (but less than $0.5$), SM is the
dominant strategy confirming the observation at the end of section \ref{sec_diff_adjust}.

\begin{figure}[h]
  \resizebox{400pt}{300pt}{\includegraphics{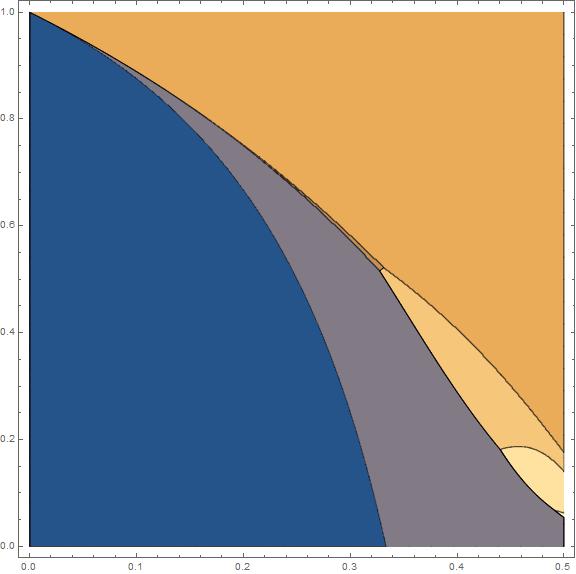}}
  \caption{TSM with $A = 2, 3, 4$ vs other strategies (HM, SM, LSM, EFSM)}
\end{figure}

\pagebreak 

\appendix

\section{Poisson processes and random walk.}\label{app_A}

 Let $N$ and $N'$ be two independent Poisson processes with parameters $\alpha$ and 
 $\alpha'$ starting at $0$: $N (0) = N' (0) = 0$. For $n, m, j, M\in \mathbb{N}$, with $m \leq j\leq M$, let 

 \begin{align*}
 \bar{S}_n &= \inf \{ t \in \mathbb{R}_{+}; N (t) + N' (t) \geq n \} \\
 \mathbf{X}_n &= (N - N') (\bar{S}_n)  \\
%
  \tau' &= \inf \{ t \in \mathbb{R}_+ ; N' (t) - N (t) = m \} \\
  \tau'' &= \inf \{ t \in \mathbb{R}_+ ; N (t) - N' (t) = M - m \}\\
  \tau &= \tau' \wedge \tau''\\
  \nu_{m, j} &= \inf \{ i \in \mathbb{N};\mathbf{X}_i = j - m \}\\ 
  \nu_m &= \nu_{m, 0} \wedge \nu_{m,M}
 \end{align*}
 

\begin{theorem}
  We have that $(\mathbf{X}_n)_{n \in \mathbb{N}}$ is a random walk with a
  probability $p = \frac{\alpha}{\alpha + \alpha'}$ (resp. $q =
  \frac{\alpha'}{\alpha + \alpha'}$) to move to the right (resp. left).
  Moreover, if $\tau_0 = \frac{1}{\alpha + \alpha'}$, we have
  \begin{align*}
    \mathbb{P} [\tau = \tau'] & = \mathbb{P} [\nu_m = \nu_{m, 0}] =  \frac{\lambda^m - \lambda^M}{1 - \lambda^M} \\
    \frac{\mathbb{E} [\tau]}{\tau_0} & = \mathbb{E} [\nu_m] =  \frac{M}{p - q}  \left( \frac{1 - \lambda^m}{1 - \lambda^M} -
    \frac{m}{M} \right) 
  \end{align*}
\end{theorem}

\begin{proof}
  For $n \in \mathbb{N}$, we have $\mathbf{X}_{n+1}=\mathbf{X}_{n}\pm 1$ 
  and $\mathbb{P}[\mathbf{X}_{n+1}=\mathbf{X}_{n}+1]=p$.
  So, $(\mathbf{X}_n)_{n \in \mathbb{N}}$ is a random walk. The two events $\{ \tau = \tau' \}$
  and $\{ \nu_m = \nu_{m, 0} \}$ are equal. So, they have the same
  probability. The computation of $\mathbb{P} [\nu_m = \nu_{m, 0}]$ can be found in
  several places. See for example \cite{F} p. $314$. Using Doob's theorem,
  we have proved in \cite{GPM2018} that $N (\tau), N' (\tau), \tau \in L^1$, $\mathbb{E} [N
  (\tau)] = \alpha \mathbb{E} [\tau]$ and $\mathbb{E} [N' (\tau)] = \alpha'
  \mathbb{E} [\tau]$. We also have:
  \begin{align*}
    \mathbb{E} [N' (\tau)] & = \mathbb{E} [N' (\tau) | \tau = \tau']\, 
    \mathbb{P} [\tau = \tau'] +\mathbb{E} [N' (\tau) | \tau = \tau'']\, 
    \mathbb{P} [\tau = \tau'']\\
    & = (m +\mathbb{E} [N (\tau) | \tau = \tau']) \, \mathbb{P} [\tau =
    \tau'] + (m - M +\mathbb{E} [N (\tau) | \tau = \tau''])\,  \mathbb{P} [\tau = \tau'']\\
    & = m\mathbb{P} [\tau = \tau'] + (m - M) \, \mathbb{P} [\tau = \tau''] +\mathbb{E} [N (\tau)]
  \end{align*}
  So, we get
  \begin{equation*}
    (\alpha - \alpha') \mathbb{E} [\tau] = (M - m) \mathbb{P} [\tau =
    \tau''] - m\mathbb{P} [\tau = \tau']
  \end{equation*}
  Therefore,
  \begin{equation*}
    \frac{\mathbb{E} [\tau]}{\tau_0} =  \frac{1}{p - q}  \left( (M - m)
    \mathbb{P} [\tau = \tau''] - m\mathbb{P} [\tau = \tau']\right ) 
    = \frac{M}{p - q}  \left( \mathbb{P} [\tau = \tau''] - \frac{m}{M} \right) 
  \end{equation*}
  and we get the result. See also \cite{F} p. $317$.
\end{proof}

\section{Proof of a result of F. Stern.}\label{rw}

We prove the last equation of Theorem \ref{en} that is a result of F.
Stern (see \cite{St}).

\medskip

\textbf{An auxiliary sequence.} We study first the sequence $(u_n)_{n \geq 0}$.

\begin{definition}\label{equndef}
  Let $(u_n)_{n \geq 0}$ be the sequence defined by induction by $u_0 = 1$ and for
  $n \geq 1$,
  \begin{equation*}
    u_n = \lambda \, \frac{1 - \lambda^n}{1 - \lambda^{n + 2}} \, u_{n - 1} +
    \frac{1}{p}\,   \frac{1 - \lambda^{n + 1}}{1 - \lambda^{n + 2}} 
  \end{equation*}
  
\end{definition}
The computation  of $u_1$ and $u_2$ gives 
\begin{align*}
\frac{u_0 + u_1}{2} &= \frac{1}{1 - pq}\\ 
\frac{u_1 + u_2}{2} &= \frac{1}{1 - 2 pq}
\end{align*}
Let $l$ be the solution to $l = \lambda l + \frac{1}{p}$, that is $l = \frac{1}{p - q}>1$, and
we have
  \begin{align*}
  l - u_n &=  \lambda \, \frac{1 - \lambda^n}{1 - \lambda^{n + 2}}  (l - u_{n -
  1}) + \frac{2}{p} \,  \frac{\lambda^{n + 1}}{1 - \lambda^{n + 2}}\\
  &\leq \lambda (l - u_{n-1}) + \frac{2}{p} \, \frac{\lambda^{n + 1}}{1 - \lambda}
  \end{align*}
Then, by induction, we have 
$$
0 \leq l - u_n \leq \lambda^n 
(l - u_0) + \frac{2 n}{p - q} \lambda^{n + 1}
$$ 
and therefore 
$$
\underset{n \rightarrow
\infty}{\lim} u_n = \frac{1}{p - q}
$$

We have a closed-form formula for $u_n$ and its partial sums.

\begin{proposition}
  For $(m, n, M) \in \mathbb{N}^3$ with $2 \leq M$ and $1 \leq m
  \leq M - 1$, we have :
  \begin{align}
    u_n & =  \frac{1 - (2 n + 3) \lambda^{n + 1} + (2 n + 3) \lambda^{n + 2}
    - \lambda^{2 n + 3}}{p (1 - \lambda)  (1 - \lambda^{n + 1})  (1 -
    \lambda^{n + 2})}  \label{closedformun}\\
    \sum_{i = M - 1 - m}^{M - 2} u_i & =  \frac{m \lambda^m - (2 M - m)
    \lambda^M + (2 M - m) \lambda^{M + m} - m \lambda^{2 M}}{p (1 - \lambda) 
    (\lambda^m - \lambda^M)  (1 - \lambda^M)}  \label{closedformpartialun}
  \end{align}
\end{proposition}

\begin{proof}
  We have that the right hand side of (\ref{closedformun}) equals to 
  $1$ when $n = 0$ and satisfies the induction from Definition \ref{equndef} when $n \in
  \mathbb{N}^{\ast}$. Hence, by induction, we get (\ref{closedformun}) for all $n$. 
  Therefore, equation (\ref{closedformpartialun}) is true for $m = 1$
  and any $M \geq 2$. Then equation (\ref{closedformpartialun}) follows by
  induction on $m$ using (\ref{closedformun}).
\end{proof}

\begin{corollary}
  Let $M \in \mathbb{N}$ with $M \geqslant 2$. We have 
  $$
  \frac{u_{M - 3} + u_{M - 2}}{2} = \frac{1}{p} \cdot \frac{1 - (M - 1) \lambda^{M - 2} + (M - 1)
  \lambda^M - \lambda^{2 M - 2}}{(1 - \lambda)  (1 - \lambda^{M - 2})  (1 -
  \lambda^M)}
  $$
\end{corollary}

\textbf{Proof of the last equation in Theorem \ref{en}.}

  Note that for any $t \in \mathbb{N}$ and  $i \in [ 0, M - 1]$ 
  $$
  \mathbb{P} [\nu_m = \nu_{m, 0} | \mathbf{X}_t = i] =
  \mathbb{P} [\nu_i = \nu_{i, 0}] = \frac{\lambda^i - \lambda^M}{1 -
  \lambda^M}
  $$ 
  is independent of $t$ and $m$. Therefore, for any $t \in
  \mathbb{N}$ and $i \in [ 1, M - 1 ]$,
  \begin{align*}
    &\mathbb{P} [(\mathbf{X}_{t + 1} = i \pm 1) | (\mathbf{X}_t = i) \wedge
    (\nu_m = \nu_{m, 0}) ]  = \frac{\mathbb{P} [(\mathbf{X}_{t +
    1} = i \pm 1) \wedge (\nu_m = \nu_{m, 0}) | \mathbf{X}_t = i
    ]}{\mathbb{P} [\nu_m = \nu_{m, 0} | \mathbf{X}_t = i
    ]}\\
    & = \mathbb{P} [(\mathbf{X}_{t + 1} = i \pm 1) | \mathbf{X}_t = i
    ] \cdot \frac{\mathbb{P} [\nu_m = \nu_{m, 0} | \mathbf{X}_{t +
    1} = i  \pm 1]}{\mathbb{P} [\nu_m = \nu_{m, 0} | \mathbf{X}_t =
    i ]}
  \end{align*}
  is also independent of $t$ and $m$. Let 
  $$
  \tilde{\mathbb{P}} (i, i \pm 1)
  = \mathbb{P} [(\mathbf{X}_{t + 1} = i \pm 1) | (\mathbf{X}_t = i)
  \wedge (\nu_m = \nu_{m, 0}) ]
  $$ 
  More precisely, we have
  \begin{equation}
    \tilde{\mathbb{P}} (i, i + 1) = p \cdot \frac{\lambda^{i + 1} -
    \lambda^M}{\lambda^i - \lambda^M}  \label{piip1}
  \end{equation}
  and
  \begin{equation}
    \tilde{\mathbb{P}} (i, i - 1) = q \cdot \frac{\lambda^{i - 1} -
    \lambda^M}{\lambda^i - \lambda^M}  \label{piim1}
  \end{equation}
  We recognize in $\tilde{\mathbb{P}}$ the $h$-Doob transform of
  $\mathbb{P}$ with
  \begin{equation*}
    h (i, j) = \frac{\lambda^j - \lambda^M}{\lambda^i - \lambda^M}
  \end{equation*}
  i.e., $\tilde{P} (i, j) = P (i, j) h (i, j)$. We check that
  \begin{equation}
    \tilde{\mathbb{P}} (i, i + 1) + \tilde{\mathbb{P}} (i, i - 1)  =  1 
    \label{s1}
  \end{equation}
  and
  \begin{equation}
    \tilde{\mathbb{P}} (M - 1, M - 2) = \tilde{\mathbb{P}} (0, 0) = 1
    \label{paa1a2}
  \end{equation}
  The probability $\tilde{\mathbb{P}}$ is the probability $\mathbb{P}$
  twisted by the condition that the hiker exits $[ 1, M - 1
  ]$ at $1$. In particular, we have for all $m \in [ 0, M ]$
  \begin{equation}
    \mathbb{E} [\nu_m | \nu_m = \nu_{m, 0} ] = \tilde{\mathbb{E}}
    [\nu_{m, 0}] = \tilde{\mathbb{E}} [\nu_m]
  \end{equation}
  where $\tilde{\mathbb{E}}$ is
  the expectation taken for the probability $\tilde{\mathbb{P}}$.

  We define for $i \in [ 0, M - 1]$, 
  $$
    v_{i, M} = \tilde{\mathbb{E}} [\nu_i] 
  $$
  and for $i {\not =} 0$ 
  \begin{equation}
    \alpha_{i, M}  =  v_{i, M} - v_{i - 1, M}  \label{defai}
  \end{equation}
  Since $v_{0, M} = 0$, we have
  \begin{equation}
    v_{i, M} = \alpha_{1, M} + \ldots + \alpha_{i, M}  \label{viM}
  \end{equation}
  All these quantities depend on $M$ since $\nu_m = \nu_{m, 0}
  \wedge \nu_{m, M}$. Note that by (\ref{paa1a2}), we have,
  \begin{equation}
    \alpha_{M - 1, M} =  1  \label{paa1}
  \end{equation}
  Moreover, by (\ref{s1}) and the Markov property, we have
  \begin{align*}
    v_{i, M} & =  \tilde{\mathbb{P}} (i, i - 1) \cdot (1 + v_{i - 1, M}) +
    \tilde{\mathbb{P}} (i, i + 1) \cdot (1 + v_{i + 1, M})\\
    & = \tilde{\mathbb{P}} (i, i - 1) \cdot v_{i - 1, M} +
    \tilde{\mathbb{P}} (i, i + 1) \cdot (v_{i, M} + \alpha_{i + 1, M}) + 1
  \end{align*}
  So, 
  \begin{equation*}
    \tilde{\mathbb{P}} (i, i - 1) \cdot v_{i, M}  =  \tilde{\mathbb{P}}
    (i, i - 1) \cdot v_{i - 1, M} + \tilde{\mathbb{P}} (i, i + 1) \cdot
    \alpha_{i + 1, M} + 1
  \end{equation*}
  and 
  \begin{equation}
    \alpha_{i, M} = \frac{\tilde{\mathbb{P}} (i, i +
    1)}{\tilde{\mathbb{P}} (i, i - 1)} \alpha_{i + 1, M} +
    \frac{1}{\tilde{\mathbb{P}} (i, i - 1)} 
  \end{equation}
  Hence we can compute $\alpha_{i, M}$ by induction on $i$ from $i = M -
  1$ to $i = 1$ with the help of (\ref{piip1}), (\ref{piim1}) and
  (\ref{paa1}). We get also $v_{i, M}$ by (\ref{viM}). Explicitly, set
  $u_{i, M} = \alpha_{M - 1 - i, M}$. Then, $u_{0, M} = 1$ and
  \begin{align*}
    u_{i, M} & =  \frac{\tilde{\mathbb{P}} (M - i - 1, M -
    i)}{\tilde{\mathbb{P}} (M - i - 1, M - i - 2)} u_{i - 1, M} +
    \frac{1}{\tilde{\mathbb{P}} (M - i - 1, M - i - 2)} \nonumber\\
    & = \frac{p \cdot \frac{\lambda^{M - i} - \lambda^M}{\lambda^{M - i -
    1} - \lambda^M}}{q \cdot \frac{\lambda^{M - i - 2} - \lambda^M}{\lambda^{M
    - i - 1} - \lambda^M}} u_{i - 1, M} + \frac{1}{q \cdot \frac{\lambda^{M -
    i - 2} - \lambda^M}{\lambda^{M - i - 1} - \lambda^M}} \nonumber\\
    & = \lambda \frac{1 - \lambda^i}{1 - \lambda^{i + 2}} u_{i - 1, M} +
    \frac{1}{p}  \frac{1 - \lambda^{i + 1}}{1 - \lambda^{i + 2}} 
  \end{align*}
  Therefore $u_{i, M}=u_i$ does not depend on $M$. Then we have by (\ref{viM})
  \begin{equation*}
    \tilde{\mathbb{E}} [\nu_m] = \sum_{i = M - 1 - m}^{M - 2} u_i 
  \end{equation*}
  and the result then follows from (\ref{closedformpartialun}).

\end{document}